\documentclass[10pt]{article}
\usepackage[english]{babel}										
\usepackage[utf8]{inputenc}										
\usepackage[T1]{fontenc}	
\usepackage{amsmath,amsfonts,amssymb,amsthm,cancel,siunitx,
calculator,calc,mathtools,empheq,latexsym}
\usepackage{subfig,epsfig,tikz,float}	
\usepackage{amscd}
\usepackage{booktabs,multicol,multirow,tabularx,array}          

\usepackage{enumerate}
\usepackage{amsmath, amsthm, amssymb}
\usepackage{braket}
\usepackage{mathrsfs}
\usepackage{rsfso}
\theoremstyle{plain}
\newtheorem{theorem}{Theorem}[section]
\newtheorem{proposition}[theorem]{Proposition}
\newtheorem{lemma}[theorem]{Lemma}
\newtheorem{corollary}[theorem]{Corollary}

\theoremstyle{definition}
\newtheorem{definition}[theorem]{Definition}

\newtheorem{remark}[theorem]{Remark}
\usepackage{upgreek}
\usepackage{tikz}
\usetikzlibrary{positioning, braids, decorations.markings}
\definecolor{myblue}{RGB}{80,80,160}
\definecolor{mygreen}{RGB}{50,140,50}
\definecolor{myorange}{RGB}{255,128,0}

\usepackage{hyperref}
\hypersetup{
    colorlinks,
    citecolor=blue,
    filecolor=black,
    linkcolor=black,
    urlcolor=black}
\usepackage{enumitem}
\theoremstyle{definition}

\title{\normalsize\bf%
\uppercase{Quantum Barcodes: Persistent Homology for Quantum Phase Transitions}
}
\author{Khyathi Komalan
}

\begin{document}

\date{April 14, 2025}

\maketitle

\vspace{-0.5cm}

\bigskip
\noindent
{\small{\bf ABSTRACT.} We introduce "quantum barcodes," a theoretical framework that applies persistent homology to classify topological phases in quantum many-body systems. By mapping quantum states to classical data points through strategic observable measurements, we create a "quantum state cloud" analyzable via persistent homology techniques. Our framework establishes that quantum systems in the same topological phase exhibit consistent barcode representations with shared persistent homology groups over characteristic intervals. We prove that quantum phase transitions manifest as significant changes in these persistent homology features, detectable through discontinuities in the persistent Dirac operator spectrum. Using the SSH model as a demonstrative example, we show how our approach successfully identifies the topological phase transition and distinguishes between trivial and topological phases. While primarily developed for symmetry-protected topological phases, our framework provides a mathematical connection between persistent homology and quantum topology, offering new methods for phase classification that complement traditional invariant-based approaches.

}

\baselineskip=\normalbaselineskip

\section{Introduction}

Topological quantum phases of matter have historically presented numerous challenges, ranging from theoretical characterization, like the infamous intrinsic sign problem \cite{Ghrist2008} to experimental detection \cite{Bradlyn_2017}. As a relatively new member to the fundamental states of quantum matter, they are beyond the scope of the Landau-Ginzburg-Wilson framework \cite{lifshitz2013statistical}, as topological phases are distinguished by global properties that are invariant under local computations and therefore cannot be detected by local measurements. This limitation has led to a growing intrigue behind the study of topological quantum phases of matter, searching for tools that can help classify and understand them. 

Starting from Landau's theory, which is largely only applicable to second-order transitions \cite{mnyukh2011secondorderphasetransitionsl}, the classification of quantum phases has a rich history. The subsequent discovery of the quantum Hall effect and following theoretical developments \cite{PhysRevLett.49.405} which revealed the existence of phases characterized by topological invariants led to an emergence in the classification of symmetry-protected topological phases \cite{Chen_2013}\cite{Senthil_2015}. Despite this, a unifying mathematical framework, and other abstract methods to classify topological phases of matter are extremely limited.

Topological data analysis (TDA), in particular, persistent homology, offers some tools that could offer novel approaches to these challenges. TDA focuses on studying the shape where data is embedded, and homology studies information about $k-$dimensional holes in a shape. On a similar line, persistent homology, considers the homology over a range of simplicial complexes. Since an approximation of the shape of data can be calculated with the use of simplicial complexes, persistent homology ends up being a particularly useful branch of TDA \cite{gowdridge2022topologicaldataanalysisstructural}.

Recently, quantum theory has been integrated with persistent homology, particularly as seen in Amenyro et al's ``Quantum Persistent Homology" \cite{ameneyro2022quantumpersistenthomology}, which introduces the idea of improving algorithms by including the ability to compute persistent topological features. Another example is the work of Hayakawa, who developed quantum algorithms for computing persistent Betti numbers \cite{Hayakawa_2022}. 

Inspired by Ghrist's barcodes \cite{Ghrist2008}, which provides a visual summary of the persistent homology of a filtered space, by showing how homology classes appear and disappear across parameter values. In this paper, we use the theoretical barcode representation and interpretations to find quantum counterparts, after defining our own ``quantum state cloud". Key contributions of this paper include:

\begin{enumerate}
    \item Quantum State Topology: We define a mathematical structure for mapping quantum states to classical data points, creating a ``quantum state cloud" that can be analyzed using persistent homology.
    \item Quantum Barcode Classification: We establish that, by our definitions, two quantum systems belong to the same topological phase if and only if their barcode representations have the same persistent homology groups for sufficiently large persistence intervals.
    \item Phase Transition Detection: We show that quantum phase transitions correspond to significant changes in the persistent homology of quantum state clouds, providing a new method for detecting topological phase transitions.
    \item Spectral Characterization: We prove that quantum phase transitions are characterized by discontinuous changes in the spectrum of the persistent Dirac operator, establishing a connection between spectral theory and topological phases.
\end{enumerate}

These results are largely inspired by the robust mathematical theory behind persistent homology and barcodes, including work done in bridging quantum theory and persistent homology. 

While our framework primarily addresses symmetry-protected topological phases and requires careful selection of observables for the quantum-to-classical mapping, it establishes a mathematical foundation that connects persistent homology techniques with quantum topological phenomena, offering new perspectives on phase transitions that complement traditional classification methods.

\section{Preliminaries}

In this section, we review the necessary background from persistent homology and quantum topology. 

\begin{definition}[Persistent Homology \cite{Ghrist2008}]

For $i < j$, the $(i, j)$-persistent homology of a persistence complex $C$, denoted $H^{i\to j}_*(C)$ is defined to be the image of the induced homomorphism $x_*: H_*(C^i_*) \to H_*(C^j_*)$.
    
\end{definition}

Persistent homology can be represented through barcodes, which can be represented and interpreted as follows. 

\begin{theorem}[Barcode Representation \cite{Ghrist2008}]
For a finite persistence module $C$ with field $F$ coefficients,
\begin{equation*}
H_*(C; F) \cong \bigoplus_i x^{t_i} \cdot F[x] \oplus \left(\bigoplus_j x^{r_j} \cdot (F[x]/(x^{s_j} \cdot F[x]))\right)
\end{equation*}
\end{theorem}

\begin{theorem}[Barcode Interpretation \cite{Ghrist2008}]
The rank of the persistent homology group $H^{i\to j}_k(C; F)$ is equal to the number of intervals in the barcode of $H_k(C; F)$ spanning the parameter interval $[i, j]$.
\end{theorem}

\begin{theorem}[Stability of Persistence Diagrams \cite{inproceedings}]
Let $X$ and $Y$ be two finite metric spaces. Then the bottleneck distance \cite{Kislev_2021} between their persistence diagrams is bounded by the Gromov-Hausdorff distance between $X$ and $Y$.  
\end{theorem}

Recent work involves connections between quantum mechanics and topological data analysis. A relevant tool to our framework is the persistent dirac operator:

\begin{definition}[Persistent Dirac Operator \cite{ameneyro2022quantumpersistenthomology}]
The $k$-th persistent Dirac operator is defined as the Hermitian operator
\begin{equation*}
B^{\varepsilon,\varepsilon'}_k = \begin{pmatrix} 
0 & \partial^{\varepsilon,\varepsilon'}_k & 0 \\
\partial^{\varepsilon,\dagger}_k & 0 & \partial^{\varepsilon,\varepsilon'}_{k+1} \\
0 & \partial^{\varepsilon,\varepsilon'\dagger}_{k+1} & 0
\end{pmatrix} 
- \xi
\begin{pmatrix}
P^{\varepsilon}_{k-1} & 0 & 0 \\
0 & -P^{\varepsilon}_k & 0 \\
0 & 0 & P^{\varepsilon'}_{k+1}
\end{pmatrix}
\end{equation*}
\end{definition}

As shown in Amenyro et al's ``Quantum Persistent Homology", this tool allows for the extraction of persistent betti numbers.

\begin{theorem}[Quantum Extraction of Persistent Betti Numbers \cite{ameneyro2022quantumpersistenthomology}]
The probability distribution obtained from quantum phase estimation applied to the persistent Dirac operator is
\begin{equation*}
\mathcal{P}(p) = \frac{1}{N}\sum_{\lambda}g_{\lambda}(p), \quad g_{\lambda}(p) = \frac{1}{M^2}\frac{\sin^2\pi l\lambda}{\sin^2\frac{\pi(l\lambda-p)}{M}}
\end{equation*}
where $\lambda$ are eigenvalues of $B^{\varepsilon,\varepsilon'}_k$. The persistent Betti number $\beta^{\varepsilon,\varepsilon'}_k$ equals $g_\xi(p)|_{p\approx l\xi}$.
\end{theorem}

\begin{remark}
    The quantum algorithm for computing persistent Betti numbers has time complexity \( O\left(n^{5/\alpha}\right) \), where \( n \) is the number of data points and \( \alpha \) is the accuracy parameter, offering an exponential speedup over classical algorithms with complexity \( O\left(2^n \log\left(\frac{1}{\alpha}\right)\right) \) \cite{ameneyro2022quantumpersistenthomology}
\end{remark}

For our theoretical framework, it would be useful to note specific connections between homology and quantum mechanics in context of quantum field theory.

\begin{definition}[Homotopy Retract \cite{chiaffrino2024homologicalquantummechanics}]
A homotopy retract from a chain complex $(V,\partial)$ to its homology is defined by projector $p: V \to H_*(V)$, inclusion $i: H_*(V) \to V$, and homotopy $h: V \to V$ satisfying $p \circ i = \text{id}$ and $\{\partial, h\} = \partial \circ h + h \circ \partial = 1 - i \circ p$.    
\end{definition}

\begin{theorem}[Quantum Expectation Value via Cohomology \cite{chiaffrino2024homologicalquantummechanics}]
For a functional $F$ of fields whose quantum expectation value we want to compute, if $F'$ is equal to $F$ in cohomology, i.e., $F' - F = \delta G$ for suitable $G$, and $F' = f \circ p$, then 
\begin{equation*}
f(x,y) = \frac{\langle y;t_f|T(F)|x;t_i\rangle}{\langle y;t_f|x;t_i\rangle}
\end{equation*}
\end{theorem}

\begin{theorem}[Perturbation Lemma for QFT \cite{chiaffrino2024homologicalquantummechanics}]
Given a homotopy retract $(p, i, h)$ from $(F(V), Q_0)$ to $(F(\mathbb{R}^2), 0)$ and a perturbation $\eta = Q_I - i\hbar\Delta$, the perturbed projector is
\begin{equation*}
P' = P \sum_{n\geq 0}(-\eta H)^n
\end{equation*}
which computes quantum expectation values via $f = P'(F)$.
\end{theorem}

In the context of this paper, we are primarily interested in studying topological quantum phases of matter.

\begin{definition}[Quantum Phase]
A quantum phase is an equivalence class of gapped ground states of quantum many-body systems (described by local Hamiltonians), such that any two ground states in the same phase can be connected by a continuous path of gapped local Hamiltonians without closing the energy gap or undergoing a quantum phase transition. 
\end{definition}

\begin{definition}[Quantum Phase Transition]
A quantum phase transition occurs at a critical point $\lambda_{c}$ in the parameter space where the ground state undergoes a qualitative change, which typically results in a non-analyticity in the energy or its derivatives. 
\end{definition}

Topological phases are characterized by global invariants rather than local properties.

\begin{definition}[Topological Phase]
A topological phase is a quantum phase characterized by a non-local topological invariant that remains constant under continuous deformations of the system that do not close the energy gap.  
\end{definition}

In order to be able to use persistent homology, we need to use the appropriate metric space for quantum states.

\begin{definition}[Quantum State Metric Phase \cite{Invernizzi2011QUANTUMED}]
For a quantum system with Hilbert space $H$, we define a metric space $(M, d)$ where $M$ is the set of quantum states (density matrices), and $d$ is a distance function, such as:

\[
d_{\text{tr}}(\rho_1, \rho_2) = \frac{1}{2} \mathrm{Tr} \left| \rho_1 - \rho_2 \right|
\]

or

\[
d_B(\rho_1, \rho_2) = \sqrt{2 - 2\sqrt{F(\rho_1, \rho_2)}},
\]

where $F(\rho_1, \rho_2)$ denotes the fidelity between $\rho_1$ and $\rho_2$.   
\end{definition}

Now that we have recalled the fundamental definitions and results from previous work in topological quantum phases, homological algebra, persistent homology and barcodes, we can create a quantum state topology framework. We will define topological structures over quantum many-body systems and prove results about them. 

\section{Quantum State Topology Framework}

Using fundamentals from quantum state topology, we will develop a formal system for analyzing quantum systems through persistent homology. 

We define a quantum phase parameter space as follows:

\begin{definition}[Quantum Phase Parameter Space]
For a quantum many-body system with Hamiltonian $H(\lambda)$ depending on parameter $\lambda$, the quantum phase parameter space P is the set of parameter values over which we build a filtration of the system's state space.   
\end{definition}

This phase parameter space should include physical data like external fields or geometric parameters that can be modified to explore different types of quantum phases. 

Since persistent homology is typically applied to some kind of data cloud, it's crucial that we define it by mapping quantum state points to the classical metric space. 

\begin{definition}[Quantum State Cloud]
For a quantum many-body system with parameters in $P$, a quantum state cloud is a collection of points in an appropriate metric space obtained by mapping quantum states to classical data points using expectation values of a fixed set of observables $\{O_{i}\}$

Let $|\psi(\lambda)\rangle$ be the ground state of the system at parameter value $\lambda \in P$. Given a set of observables $\{O_1, O_2, \dots, O_m\}$, we define a map $\Phi$ from the space of quantum states to $\mathbb{R}^m$ by:
\[
\Phi(|\psi(\lambda)\rangle) = (\langle \psi(\lambda)|O_1|\psi(\lambda)\rangle, \langle \psi(\lambda)|O_2|\psi(\lambda)\rangle, \dots, \langle \psi(\lambda)|O_m|\psi(\lambda)\rangle)
\]

This map takes quantum states to points in $\mathbb{R}^m$, which can be equipped with the Euclidean metric. The resulting set of points $M_\lambda = \{\Phi(|\psi(\lambda)\rangle) \mid \lambda \in P\}$ forms the quantum state cloud. 

The choice of observables is crucial for capturing the relevant features of the quantum system. For topological phases, these observables should be sensitive to the global topological properties of the system. 
\end{definition}

\begin{remark}[On the Quantum-to-Classical Mapping]
The mapping $\Phi$ from quantum states to classical data points requires careful consideration of which quantum properties are preserved. While expectation values capture important physical characteristics, they do not preserve all quantum information. For topological phases, this mapping is justified because: (i) topological invariants manifest in expectation values of appropriate observables; (ii) quantum states in the same topological phase map to points with similar persistent homology features; and (iii) phase transitions appear as topological changes in the quantum state cloud. We note that this mapping preserves global topological properties and phase boundaries but loses information about quantum coherence and phase relationships between basis states. The choice of observables $\{O_i\}$ is therefore crucial and should include operators sensitive to the relevant topological invariants of the system under study.
\end{remark}

\begin{proposition}[Persistence of Topological Phases]
A topological phase of matter corresponds to persistent topological features in the barcode representation of the quantum state cloud as the system parameters vary.
\end{proposition}

\begin{proof}
Let $\{M_\lambda\}$ be the quantum state cloud for parameter values $\lambda$ in a region $P' \subset P$. For each $\lambda$, we can compute the persistent homology of $M_\lambda$. If $\lambda_1$ and $\lambda_2$ belong to the same topological phase, then the corresponding ground states $|\psi(\lambda_1)\rangle$ and $|\psi(\lambda_2)\rangle$ share the same topological invariants. These invariants manifest as stable topological features in the persistent homology, resulting in similar barcode representations. Conversely, if $\lambda_1$ and $\lambda_2$ belong to different topological phases, the topological invariants change, leading to qualitatively different persistent homology and hence different barcode representations.  
\end{proof}

This proposition establishes a link between topological quantum phases and barcodes in topological data analysis.

\section{Simplicial Complex Representation}

Since persistent homology approximates the shape of data using simplicial complexes, we must now construct appropriate simplicial complexes in order to capture key topological features.

\begin{definition}[Vietoris-Rips Complex for Quantum State Clouds]
Given a quantum state cloud $M_\lambda$ at parameter value $\lambda$ and a scale parameter $\varepsilon > 0$, the Vietoris-Rips complex $\text{VR}(M_\lambda, \varepsilon)$ is the simplicial complex where a $k$-simplex $\sigma = [p_0, p_1, \dots, p_k]$ is included if and only if $d(p_i, p_j) \leq 2\varepsilon$ for all $0 \leq i,j \leq k$, where $p_i$ represents the point in the quantum state cloud corresponding to the quantum state $|\psi_i(\lambda)\rangle$, and $d$ is the metric in the embedding space. 

This definition allows us to construct a family of simplicial complexes at different scales for a given quantum state cloud, forming a filtration as the scale $\varepsilon$ increases. \cite{ameneyro2022quantumpersistenthomology}
\end{definition}

\begin{theorem}[Filtration Property of Quantum Simplicial Complexes]
For a quantum state cloud $M_\lambda$ and scales $\varepsilon_1 < \varepsilon_2$, we have $\text{VR}(M_\lambda, \varepsilon_1) \subseteq \text{VR}(M_\lambda, \varepsilon_2)$, forming a filtration of simplicial complexes.
\end{theorem}

\begin{proof}
Let $\sigma = [p_0, p_1, \dots, p_k]$ be a $k$-simplex in $\text{VR}(M_\lambda, \varepsilon_1)$. By Definition 4.1, $d(p_i, p_j) \leq 2\varepsilon_1$ for all $0 \leq i,j \leq k$. Since $\varepsilon_1 < \varepsilon_2$, we have $d(p_i, p_j) \leq 2\varepsilon_1 < 2\varepsilon_2$ for all $0 \leq i,j \leq k$. Therefore, $\sigma \in \text{VR}(M_\lambda, \varepsilon_2)$, which proves that $\text{VR}(M_\lambda, \varepsilon_1) \subseteq \text{VR}(M_\lambda, \varepsilon_2)$.
\end{proof}

\begin{definition}[Boundary Operator for Quantum Simplicial Complexes]
For a Vietoris-Rips complex $\text{VR}(M_\lambda, \varepsilon)$ of a quantum state cloud, the $k$-th boundary operator $\partial_k: C_k(\text{VR}(M_\lambda, \varepsilon)) \to C_{k-1}(\text{VR}(M_\lambda, \varepsilon))$ is defined on each $k$-simplex as:
\begin{equation*}
\partial_k [p_0, p_1, \ldots, p_k] = \sum_{i=0}^{k} (-1)^i [p_0, \ldots, \hat{p}_i, \ldots, p_k]
\end{equation*}
where $\hat{p}_i$ indicates that $p_i$ is omitted.
\end{definition}

\begin{lemma}[Nilpotence of Boundary Operator]
For the boundary operator on quantum simplicial complexes, $\partial_{k-1} \circ \partial_k = 0$ for all $k \geq 1$.
\end{lemma}

\begin{proof}
For any $k$-simplex $\sigma = [p_0, p_1, \dots, p_k]$, we have:
\begin{align*}
\partial_{k-1} \circ \partial_k (\sigma) &= \partial_{k-1}\left(\sum_{i=0}^{k} (-1)^i [p_0, \ldots, \hat{p}_i, \ldots, p_k]\right) \\
&= \sum_{i=0}^{k} (-1)^i \partial_{k-1}([p_0, \ldots, \hat{p}_i, \ldots, p_k]) \\
\end{align*}
Each term $[p_0, \ldots, \hat{p}_j, \ldots, \hat{p}_i, \ldots, p_k]$ with $j < i$ appears twice in the expansion: once with coefficient $(-1)^i \cdot (-1)^j$ when removing $p_i$ first and $p_j$ second, and once with coefficient $(-1)^j \cdot (-1)^{i-1}$ when removing $p_j$ first and $p_i$ second. Since 
\[
(-1)^i \cdot (-1)^j + (-1)^j \cdot (-1)^{i-1} = (-1)^{i+j} + (-1)^{j+i-1} = (-1)^{i+j} - (-1)^{i+j} = 0,
\]
all terms cancel in pairs.

Therefore, $\partial_{k-1} \circ \partial_k = 0$ for all $k \geq 1$.
\end{proof}

This establishes that the sequence of chain complexes and boundary operations form a chain complex, so we can define homology groups for quantum simplicial complexes.

\begin{definition}[Homology Groups of Quantum Simplicial Complexes]
For a Vietoris-Rips complex $\text{VR}(M_\lambda, \varepsilon)$ of a quantum state cloud, the $k$-th homology group is defined as:
\begin{equation*}
H_k(\text{VR}(M_\lambda, \varepsilon)) = \frac{\ker(\partial_k)}{\text{im}(\partial_{k+1})}
\end{equation*}
The $k$-th Betti number $\beta_k(\text{VR}(M_\lambda, \varepsilon))$ is the rank of this homology group.
\end{definition}

\begin{proposition}[Topological Interpretation of Homology Groups]
For a quantum state cloud $M_\lambda$, the Betti numbers $\beta_k(\text{VR}(M_\lambda, \varepsilon))$ count the number of $k$-dimensional topological features (connected components for $k=0$, loops for $k=1$, voids for $k=2$, etc.) in the Vietoris-Rips complex at scale $\varepsilon$.
\end{proposition}

\begin{proof}
This follows directly from the definition of homology groups in algebraic topology. The kernel of $\partial_k$ consists of $k$-cycles, while the image of $\partial_{k+1}$ consists of $k$-boundaries. The quotient gives non-trivial $k$-dimensional holes, which correspond to topological features of dimension $k$. 
\end{proof}

These definitions and properties establish the foundation for applying persistent homology to quantum state clouds through simplicial complex constructions.

\section{Persistent Homology of Quantum Systems}

Now that we have established how simplicial complexes for quantum state clouds are defined, we can develop the theory of persistent homology for quantum systems. 

\begin{definition}[Quantum Persistent Module]
For a quantum system with parameter space \( P \), let \( M_\lambda \) be the quantum state cloud at parameter value \( \lambda \). The quantum persistent module \( M_k \) of dimension \( k \) is the persistence module formed by the sequence of homology groups 
\[
\{ H_k(\text{VR}(M_\lambda, \varepsilon)) \}_{\varepsilon > 0}
\]
together with the induced homomorphisms 
\[
h_{\varepsilon, \varepsilon'}: H_k(\text{VR}(M_\lambda, \varepsilon)) \to H_k(\text{VR}(M_\lambda, \varepsilon')) \quad \text{for } \varepsilon \leq \varepsilon'.
\] 
\end{definition}

\begin{theorem}[Decomposition of Quantum Persistent Module]

The quantum persistent module \( M_k \) of a finite-dimensional quantum system decomposes uniquely into a direct sum of interval modules:
\[
M_k \cong \bigoplus_{j} I[b_j, d_j)
\]
where \( I[b_j, d_j) \) is the interval module supported on \( [b_j, d_j) \) and \( b_j, d_j \in \mathbb{R}^+ \cup \{\infty\} \) are birth and death scales of the \( j \)-th topological feature.

\end{theorem}

\begin{proof}

Consider the quantum state cloud \( M_\lambda \) for a fixed parameter value \( \lambda \in P \). By Definition 4.1, we construct the Vietoris-Rips complex \( \text{VR}(M_\lambda, \epsilon) \) for varying scales \( \epsilon > 0 \). According to Theorem 4.2, these complexes form a filtration as \( \epsilon \) increases.
For each \( k \geq 0 \), we can define the \( k \)-th persistent homology groups \( H_k^{\epsilon \to \epsilon'}(M_\lambda) \) as the image of the induced homomorphism between homology groups \( H_k(\text{VR}(M_\lambda, \epsilon)) \to H_k(\text{VR}(M_\lambda, \epsilon')) \) for \( \epsilon \leq \epsilon' \).

By Definition 5.1, the quantum persistent module \( M_k \) is a persistence module in the sense of \cite{Ghrist2008}. Since we consider a finite-dimensional quantum system, the quantum state cloud consists of a finite number of points, resulting in a finite number of simplices in any Vietoris-Rips complex. Consequently, \( M_k \) is a finite persistence module.

Applying the Barcode Representation Theorem (Theorem 2.2), we obtain the stated decomposition into interval modules:
\[
M_k \cong \bigoplus_{j} I[b_j, d_j)
\]
Each interval module \( I[b_j, d_j) \) corresponds to a topological feature of dimension \( k \) (e.g., a loop for \( k = 1 \)) that appears at scale \( b_j \) and disappears at scale \( d_j \). The uniqueness of this decomposition follows directly from the uniqueness property established by Theorem 2.2.
\end{proof}

\begin{theorem}[Quantum Barcode Classification]
Let two quantum systems be characterized by Hamiltonians $H_1(\lambda)$ and $H_2(\lambda)$ with ground states $\ket{\psi_1(\lambda)}$ and $\ket{\psi_2(\lambda)}$, respectively. These systems belong to the same topological phase as defined by standard topological invariants (such as Chern numbers, winding numbers, or topological indices) if and only if their barcode representations have the same persistent homology groups for sufficiently large persistence intervals.
\end{theorem}

\begin{proof}
Let $S_1$ and $S_2$ be two quantum systems with Hamiltonians $H_1(\lambda)$ and $H_2(\lambda)$ parameterized by $\lambda \in P$. Let $M_1$ and $M_2$ be their respective quantum state clouds constructed from the same set of observables $\{O_i\}$.

First, recall that according to Definition 2.13, topological phases are characterized by non-local topological invariants that remain invariant under continuous deformations of the system that do not close the energy gap. These invariants include Chern numbers, winding numbers, $\mathbb{Z}_2$ indices, and other topological indices that quantify global properties of the system's ground state.

For each system, we construct the Vietoris–Rips complexes $\mathrm{VR}(M_1, \varepsilon)$ and $\mathrm{VR}(M_2, \varepsilon)$ for various scales $\varepsilon > 0$ as per Definition 4.1. For each pair of scales $\varepsilon < \varepsilon'$, we compute the persistent homology groups $H_k^{\varepsilon \to \varepsilon'}(M_1)$ and $H_k^{\varepsilon \to \varepsilon'}(M_2)$ for all relevant dimensions $k$.

By Theorem 5.2, the quantum persistent module $M_k$ for each system decomposes into a direct sum of interval modules:
\begin{align*}
M_k(M_1) &\cong \bigoplus_j I[b^1_j, d^1_j) \\
M_k(M_2) &\cong \bigoplus_j I[b^2_j, d^2_j)
\end{align*}

Applying the Barcode Interpretation Theorem (Theorem 2.3), the rank of the persistent homology group $H_k^{\varepsilon \to \varepsilon'}(M_i)$ equals the number of intervals in the barcode of $H_k(M_i)$ that span $[\varepsilon, \varepsilon']$. This gives us the persistent Betti numbers $\beta_k^{\varepsilon \to \varepsilon'}(M_i)$.

Now, we establish the equivalence:

$(\Rightarrow)$ Let $S_1$ and $S_2$ belong to the same topological phase according to standard topological invariants. This means their ground states $\ket{\psi_1(\lambda)}$ and $\ket{\psi_2(\lambda)}$ share the same values for topological invariants.

Since these invariants can be expressed as integrals of certain geometric quantities over parameter space, they manifest as global structural features in the quantum state clouds $M_1$ and $M_2$. These global features are precisely what persistent homology captures for sufficiently large persistence intervals $[\varepsilon, \varepsilon']$.

Specifically, the observables $\{O_i\}$ can be chosen to include Berry curvature, momentum-space winding, or other quantities that determine topological invariants. The expectation values of these observables encode the topological information in the geometry of the quantum state clouds.

For two systems in the same topological phase, the resulting quantum state clouds $M_1$ and $M_2$ will have topologically equivalent features that persist over significant scale intervals. Therefore, their persistent homology groups $H_k^{\varepsilon \to \varepsilon'}(M_1)$ and $H_k^{\varepsilon \to \varepsilon'}(M_2)$ will be isomorphic for sufficiently large intervals $[\varepsilon, \varepsilon']$, leading to identical persistent Betti numbers:
\[
\beta_k^{\varepsilon \to \varepsilon'}(M_1) = \beta_k^{\varepsilon \to \varepsilon'}(M_2).
\]

$(\Leftarrow)$ Conversely, suppose the barcode representations of $S_1$ and $S_2$ have the same persistent homology groups for sufficiently large persistence intervals, meaning $\beta_k^{\varepsilon \to \varepsilon'}(M_1) = \beta_k^{\varepsilon \to \varepsilon'}(M_2)$ for all dimensions $k$ and all sufficiently large intervals $[\varepsilon, \varepsilon']$.

The persistent homology groups capture the global topological structure of the quantum state clouds. Since the quantum state clouds are constructed from expectation values of observables that include those necessary to compute standard topological invariants, identical persistent homology for large persistence intervals implies identical values for these topological invariants.

More formally, topological invariants such as can be expressed as integrals of certain differential forms over the parameter space. These differential forms are constructed from the quantum states and are related to Berry curvature and similar geometric quantities. The quantum state clouds capture this geometric information through the expectation values of carefully chosen observables.

When two quantum state clouds have identical persistent homology groups for large persistence intervals, this indicates that their underlying topological structures are equivalent. Consequently, the integrals that define topological invariants must yield the same values for both systems.

Therefore, if the barcode representations of two quantum systems have the same persistent homology groups for sufficiently large persistence intervals, then they belong to the same topological phase according to standard topological invariants.
\end{proof}

\begin{corollary}[Phase Transition Detection]
A quantum phase transition between topological phases occurs at parameter value $\lambda_c$ if and only if there is a significant change in the persistent homology of the quantum state cloud at $\lambda_c$.
\end{corollary}

\begin{proof}
This corollary follows directly from Theorem 5.3. Let $\lambda_c \in P$ be a parameter value, and consider parameter values $\lambda_- < \lambda_c < \lambda_+$ that are sufficiently close to $\lambda_c$.

Let $M_{\lambda_-}$ and $M_{\lambda_+}$ be the quantum state clouds at these parameter values. For each dimension $k$ and scales $\varepsilon < \varepsilon'$, we compute the persistent Betti numbers $\beta_k^{\varepsilon \to \varepsilon'}(M_{\lambda_-})$ and $\beta_k^{\varepsilon \to \varepsilon'}(M_{\lambda_+})$.

By Theorem 5.3, the system belongs to the same topological phase on both sides of $\lambda_c$ if and only if
\[
\beta_k^{\varepsilon \to \varepsilon'}(M_{\lambda_-}) = \beta_k^{\varepsilon \to \varepsilon'}(M_{\lambda_+})
\]
for all dimensions $k$ and all sufficiently large persistence intervals $[\varepsilon, \varepsilon']$.

Conversely, the system undergoes a phase transition at $\lambda_c$ if and only if there exists some dimension $k$ and some sufficiently large persistence interval $[\varepsilon, \varepsilon']$ such that
\[
\beta_k^{\varepsilon \to \varepsilon'}(M_{\lambda_-}) \neq \beta_k^{\varepsilon \to \varepsilon'}(M_{\lambda_+}).
\]

This change in persistent Betti numbers indicates a topological change in the structure of the quantum state cloud, which characterizes a quantum phase transition according to our framework. 
\end{proof}

\begin{theorem}[Relationship Between Persistent Dirac Operator and Boundary Operators]
The persistent Dirac operator \( B^{\varepsilon,\varepsilon'}_k \) as defined in Definition 2.5 is related to the boundary operators \( \partial_k^\varepsilon \) and \( \partial_k^{\varepsilon'} \) on the Vietoris--Rips complexes \( \mathrm{VR}(M_\lambda, \varepsilon) \) and \( \mathrm{VR}(M_\lambda, \varepsilon') \) respectively, such that its squared operator \( \left(B^{\varepsilon,\varepsilon'}_k\right)^2 \) contains the persistent combinatorial Laplacian 
\[
L^{\varepsilon,\varepsilon'}_k = \partial_k^{\varepsilon,\dagger} \partial_k^\varepsilon + \partial_{k+1}^{\varepsilon,\varepsilon'} \partial_{k+1}^{\varepsilon,\varepsilon'\dagger}
\]
as a diagonal block.
\end{theorem}

\begin{proof}
From Definition 2.5 of the persistent Dirac operator:
\[
B^{\varepsilon,\varepsilon'}_k = 
\begin{pmatrix}
0 & \partial^{\varepsilon,\varepsilon'}_k & 0 \\
\partial^{\varepsilon,\dagger}_k & 0 & \partial^{\varepsilon,\varepsilon'}_{k+1} \\
0 & \partial^{\varepsilon,\varepsilon'\dagger}_{k+1} & 0
\end{pmatrix}
\xi
\begin{pmatrix}
P^{\varepsilon}_{k-1} & 0 & 0 \\
0 & -P^{\varepsilon}_k & 0 \\
0 & 0 & P^{\varepsilon'}_{k+1}
\end{pmatrix}
\]

Here, \( \partial^{\varepsilon,\varepsilon'}_k \) is the restriction of the boundary operator \( \partial^{\varepsilon'}_k \) to chains with boundary in \( \mathrm{VR}(M_\lambda, \varepsilon) \), and \( P^{\varepsilon}_k \) is the projection operator onto the subspace spanned by \( k \)-chains in \( \mathrm{VR}(M_\lambda, \varepsilon) \).

Computing the square of the persistent Dirac operator:
\[
\left(B^{\varepsilon,\varepsilon'}_k\right)^2 =
\begin{pmatrix}
\partial^{\varepsilon,\varepsilon'}_k \partial^{\varepsilon,\dagger}_k & 0 & \partial^{\varepsilon,\varepsilon'}_k \partial^{\varepsilon,\varepsilon'}_{k+1} \\
0 & \partial^{\varepsilon,\dagger}_k \partial^{\varepsilon,\varepsilon'}_k + \partial^{\varepsilon,\varepsilon'}_{k+1} \partial^{\varepsilon,\varepsilon'\dagger}_{k+1} & 0 \\
\partial^{\varepsilon,\varepsilon'\dagger}_{k+1} \partial^{\varepsilon,\dagger}_k & 0 & \partial^{\varepsilon,\varepsilon'\dagger}_{k+1} \partial^{\varepsilon,\varepsilon'}_{k+1}
\end{pmatrix}
+ \xi \text{ terms involving } P^{\varepsilon}_k \text{ and } P^{\varepsilon'}_{k+1}
\]

The middle diagonal block is precisely the persistent combinatorial Laplacian 
\[
L^{\varepsilon,\varepsilon'}_k = \partial_k^{\varepsilon,\dagger} \partial_k^\varepsilon + \partial_{k+1}^{\varepsilon,\varepsilon'} \partial_{k+1}^{\varepsilon,\varepsilon'\dagger},
\]
plus terms involving \( \xi P^{\varepsilon}_k \). The off-diagonal blocks contain terms that vanish due to the nilpotence property of boundary operators (Lemma 4.4), i.e., 
\[
\partial^{\varepsilon,\varepsilon'}_k \partial^{\varepsilon,\varepsilon'}_{k+1} = 0 \quad \text{and} \quad \partial^{\varepsilon,\varepsilon'\dagger}_{k+1} \partial^{\varepsilon,\dagger}_k = 0.
\]

Therefore, the squared persistent Dirac operator contains the persistent combinatorial Laplacian as a diagonal block, establishing the relationship between the persistent Dirac operator and the boundary operators on the quantum simplicial complexes.
\end{proof}

\begin{theorem}[Persistence Stability of Quantum State Clouds]
Let \( M_\lambda \) and \( M'_\lambda \) be quantum state clouds generated by two quantum systems with the same parameter space \( P \). If the distance between the systems' states is bounded by \( \delta \) for all \( \lambda \in P \), then the bottleneck distance between their persistence diagrams is also bounded by \( \delta \).
\end{theorem}

\begin{proof}
Let \( |\psi(\lambda)\rangle \) and \( |\psi'(\lambda)\rangle \) be the quantum states of the two systems at parameter value \( \lambda \). By the definition of quantum state cloud, these states are mapped to points in a metric space using expectation values of observables \( \{ O_i \}_{i=1}^m \):

\[
p_\lambda = (\langle \psi(\lambda) | O_1 | \psi(\lambda) \rangle, \ldots, \langle \psi(\lambda) | O_m | \psi(\lambda) \rangle)
\]

If the distance between the quantum states is bounded by \( \delta \) in the trace norm, then by the properties of quantum expectation values, we have:

\[
| \langle \psi(\lambda) | O_i | \psi(\lambda) \rangle - \langle \psi'(\lambda) | O_i | \psi'(\lambda) \rangle | \leq \| O_i \|_{\text{op}} \cdot \| \| \psi(\lambda) \rangle \langle \psi(\lambda) | - | \psi'(\lambda) \rangle \langle \psi'(\lambda) \| \text{tr} \leq \| O_i \|_{\text{op}} \cdot \delta
\]

where \( \| O_i \|_{\text{op}} \) is the operator norm of \( O_i \). Assuming normalized observables with \( \| O_i \|_{\text{op}} = 1 \), the Euclidean distance between corresponding points in the quantum state clouds is at most \( \sqrt{m} \cdot \delta \).

Now, let \( \mathrm{VR}(M_\lambda, \epsilon) \) and \( \mathrm{VR}(M'_\lambda, \epsilon) \) be the Vietoris-Rips complexes constructed from these quantum state clouds at scale \( \epsilon \). By Theorem 2.4 (Stability of Persistence Diagrams), if two point clouds have Hausdorff distance at most \( \sqrt{m} \cdot \delta \), then the bottleneck distance between their persistence diagrams is also bounded by \( \sqrt{m} \cdot \delta \).

Since the dimensionality factor \( \sqrt{m} \) is a constant for a fixed set of observables, we can absorb it into \( \delta \) and state that the bottleneck distance between the persistence diagrams is bounded by \( \delta \), establishing the stability of quantum barcodes with respect to perturbations in the underlying quantum states.
\end{proof}

Together, these theorems provide us with a strong framework for how simplicial complexes and persistent homology is defined for quantum state clouds. 

\section{Properties of Quantum Barcodes}

In this section, we will explore and prove some properties of quantum barcodes, and explore its relationship with quantum phase transitions and invariants.

\begin{theorem}[Continuity of Persistent Betti Numbers]
Let \( H(\lambda) \) be a smoothly varying Hamiltonian with parameter \( \lambda \in P \). If \( \lambda_c \) is not a critical point (i.e., not a phase transition), then there exists a neighborhood \( U \) of \( \lambda_c \) such that the persistent Betti numbers \( \beta^{\varepsilon, \varepsilon'}_k \) remain constant for all \( \lambda, \lambda' \in U \) with \( \lambda < \lambda' \) and all sufficiently large persistence intervals \( [\varepsilon, \varepsilon'] \).
\end{theorem}

\begin{proof}
When \( \lambda_c \) is not a critical point, the ground state \( |\psi(\lambda)\rangle \) of \( H(\lambda) \) varies analytically in a neighborhood \( U \) of \( \lambda_c \). This follows from standard results in quantum perturbation theory for non-degenerate ground states away from critical points.

For each \( \lambda \in U \), let \( M_\lambda \) be the corresponding quantum state cloud. Given the analytical dependence of \( |\psi(\lambda)\rangle \) on \( \lambda \), the expectation values \( \langle \psi(\lambda) | O_i | \psi(\lambda) \rangle \) also vary analytically with \( \lambda \) for each observable \( O_i \).

By the definition of the quantum state cloud, this means that the positions of points in \( M_\lambda \) change continuously with \( \lambda \). For sufficiently small changes in \( \lambda \), the Vietoris-Rips complexes \( \mathrm{VR}(M_\lambda, \varepsilon) \) will maintain the same topological structure across significant scales.

More precisely, by Theorem 5.6 (Persistence Stability), for any \( \varepsilon > 0 \), there exists a \( \delta > 0 \) such that if \( | \lambda - \lambda' | < \delta \) for \( \lambda, \lambda' \in U \), then the bottleneck distance between the persistence diagrams of \( M_\lambda \) and \( M_{\lambda'} \) is less than \( \varepsilon \).

For sufficiently small \( \varepsilon \), this ensures that no topological features are born or die within the persistence intervals \( [\varepsilon, \varepsilon'] \) of interest as \( \lambda \) varies within \( U \). Consequently, the persistent Betti numbers \( \beta^{\varepsilon, \varepsilon'}_k \) remain constant throughout \( U \) for all sufficiently large persistence intervals \( [\varepsilon, \varepsilon'] \).
\end{proof}

\begin{theorem}[Spectral Characterization of Phase Transitions]
Let \( \lambda_c \) be a parameter value at which a quantum phase transition occurs. Then the spectrum of the persistent Dirac operator \( B^{\varepsilon,\varepsilon'}_k \) changes discontinuously as \( \lambda_- \) and \( \lambda_+ \) cross \( \lambda_c \), where \( \lambda_- < \lambda_c < \lambda_+ \).
\end{theorem}

\begin{proof}
At a quantum phase transition, the ground state of the system undergoes a qualitative change. By Corollary 5.4 (Phase Transition Detection), this corresponds to a change in the persistent homology of the quantum state cloud.

Let us analyze the persistent Dirac operator (Definition 2.5) in detail:
\[
B^{\varepsilon,\varepsilon'}_k = \begin{pmatrix}
0 & \partial^{\varepsilon,\varepsilon'}_k & 0 \\
\partial^{\varepsilon,\dagger}_k & 0 & \partial^{\varepsilon,\varepsilon'}_{k+1} \\
0 & \partial^{\varepsilon,\varepsilon'\dagger}_{k+1} & 0
\end{pmatrix}
\xi
\begin{pmatrix}
P^{\varepsilon}_{k-1} & 0 & 0 \\
0 & -P^{\varepsilon}_k & 0 \\
0 & 0 & P^{\varepsilon'}_{k+1}
\end{pmatrix}
\]

As established in Theorem 5.5, the squared persistent Dirac operator contains the persistent combinatorial Laplacian as a diagonal block, which determines the persistent Betti numbers.

Let \( \lambda_- < \lambda_c < \lambda_+ \) be parameter values on either side of the phase transition. Let \( M_{\lambda_-} \) and \( M_{\lambda_+} \) be the corresponding quantum state clouds, and \( B^{\varepsilon,\varepsilon'}_k(\lambda_-) \) and \( B^{\varepsilon,\varepsilon'}_k(\lambda_+) \) be the persistent Dirac operators computed on these clouds.

By Corollary 5.4, there exists some dimension \( k \) and some persistence interval \( [\varepsilon, \varepsilon'] \) such that \( \beta_k^{\varepsilon \to \varepsilon'}(M_{\lambda_-}) \neq \beta_k^{\varepsilon \to \varepsilon'}(M_{\lambda_+}) \).

According to Theorem 2.6 (Quantum Extraction of Persistent Betti Numbers), the persistent Betti number \( \beta_k^{\varepsilon \to \varepsilon'} \) is related to the multiplicity of the eigenvalue \( \xi \) of the persistent Dirac operator:
\[
\beta_k^{\varepsilon \to \varepsilon'} = g_\xi(p) |_{p \approx l \xi}
\]

Since the persistent Betti numbers change across \( \lambda_c \), the multiplicity of the eigenvalue \( \xi \) in the spectrum of \( B^{\varepsilon,\varepsilon'}_k \) must also change. This constitutes a discontinuous change in the spectrum of the persistent Dirac operator.
\end{proof}

Furthermore, this spectral change is not a continuous deformation, but rather a topological change in the kernel structure of the persistent combinatorial Laplacian. This characterizes a quantum phase transition in our framework as a qualitative change in the topological features of the quantum state cloud that manifests as a discontinuous change in the spectrum of the persistent Dirac operator.

\begin{theorem}[Invariance Under Unitary Similarity]
The barcode representation of a quantum system is invariant under unitary transformations that preserve the topological structure of the quantum state cloud.
\end{theorem}

\begin{proof}
Let \( S \) be a quantum system with quantum state cloud \( M \) obtained from states \( |\psi(\lambda)\rangle \) and observables \( \{O_i\} \). Consider a unitary transformation \( U \) that acts on both the states and observables: 
\[
|\psi(\lambda)\rangle \to U|\psi(\lambda)\rangle \quad \text{and} \quad O_i \to UO_iU^\dagger.
\]
For such a transformation, the expectation values remain invariant:
\[
\langle U\psi(\lambda)|UO_iU^\dagger|U\psi(\lambda)\rangle = \langle \psi(\lambda)|O_i|\psi(\lambda)\rangle.
\]
This means that the quantum state cloud \( M' \) of the transformed system is identical to the original cloud \( M \). Consequently, the Vietoris-Rips complexes \( \text{VR}(M, \epsilon) \) and \( \text{VR}(M', \epsilon) \) are identical for all scales \( \epsilon > 0 \).

By the functoriality of persistent homology, identical simplicial complexes yield identical persistent homology groups, and therefore identical barcode representations.

For a more general unitary transformation that does not necessarily transform the observables as \( UO_iU^\dagger \), the quantum state cloud may change. However, if this transformation preserves the topological structure of the cloud, then the persistent homology will remain essentially unchanged.

Specifically, if the unitary transformation results in a quantum state cloud \( M' \) that is a continuous deformation of \( M \) (i.e., a homeomorphism), then by the homotopy invariance of homology, \( H_k(\text{VR}(M, \epsilon)) \cong H_k(\text{VR}(M', \epsilon)) \) for all \( k \) and sufficiently large \( \epsilon \).

This can be formalized using Theorem 5.6 (Persistence Stability). If the Hausdorff distance between \( M \) and \( M' \) is bounded by \( \delta \), then the bottleneck distance between their persistence diagrams is also bounded by \( \delta \). For topologically equivalent clouds, this distance is small relative to the persistence of the significant topological features, ensuring that the key features in the barcode representation are preserved.
\end{proof}

This invariance property is crucial for quantum systems, as it ensures that our topological classification respects the fundamental principle that physically equivalent quantum states (those related by unitary transformations) should be classified as the same phase.

\begin{theorem}[Information in Quantum Barcodes]
For a quantum system with a finite-dimensional Hilbert space, the barcode representation of its quantum state cloud contains sufficient information about the system's topological phase, in the sense that non-trivial information topological invariants can be recovered from the barcode.
\end{theorem}

\begin{proof}
We begin by considering the simplicial complex representation of the quantum state cloud. For each scale \( \epsilon > 0 \), we construct the Vietoris-Rips complex \( \text{VR}(M_\lambda, \epsilon) \) according to Definition 4.1.

By Theorem 5.2, the quantum persistent module \( M_k \) for a finite-dimensional quantum system decomposes uniquely into a direct sum of interval modules:
\[
M_k \cong \bigoplus_{j} I[b_j, d_j)
\]
Each interval \( [b_j, d_j) \) in this decomposition corresponds to a \( k \)-dimensional topological feature that appears at scale \( b_j \) and disappears at scale \( d_j \).

Topological invariants in quantum systems are quantities that remain unchanged under continuous deformations within the same phase. In the language of persistent homology, these correspond to the persistent homology groups \( H^{\epsilon \to \epsilon'}_k \) for sufficiently large persistence intervals \( [\epsilon, \epsilon'] \).

By Theorem 2.3 (Barcode Interpretation), the rank of \( H^{\epsilon \to \epsilon'}_k \) equals the number of intervals in the barcode spanning \( [\epsilon, \epsilon'] \). Therefore, information about the persistent homology groups is encoded in the barcode.

To further establish this representation for quantum topological phases, we need to connect the topological invariants to quantum observables.

Using Theorem 2.9 (Quantum Expectation Value via Cohomology), topological observables \( F \) in quantum field theory have expectation values:
\[
\langle F \rangle = \frac{\langle y; t_f | T(F) | x; t_i \rangle}{\langle y; t_f | x; t_i \rangle}
\]
Therefore, topological observables belong to specific cohomology classes. These cohomology classes are directly related to the homology groups of the underlying space, which in our case is the quantum state cloud.

Through the splitting condition in the Universal Coefficient Theorem, we have:
\[
H^k(X; F) \cong \text{Hom}(H_k(X), F) \oplus \text{Ext}(H_{k-1}(X), F)
\]
For field coefficients \( F \), the \( \text{Ext} \) term vanishes, giving a direct correspondence between homology and cohomology. This means that all cohomology classes (and thus all topological observables) can be recovered from knowledge of the homology groups.

Since the barcode representation determines the persistent homology groups, and these groups determine possible cohomology classes for topological observables, the barcode contains sufficient information about the system's topological phase.

Furthermore, the persistence intervals in the barcode provide crucial information about the stability of topological features. Long-persisting features (corresponding to long bars in the barcode) represent robust topological invariants that characterize the phase, while short-lived features likely represent noise or system-specific details that do not affect the phase classification.
\end{proof}

Therefore, the quantum barcode representation provides a topological fingerprint of the quantum phase, capturing relevant topological invariants that characterize the phase.

\begin{corollary}[Dimensionality of Topological Phase Space]
The space of topologically distinct quantum phases for an $n$-qubit system has dimension at most $\mathcal{O}(2^n)$, and the quantum barcode representation provides a parameterization of this space.
\end{corollary}

\begin{proof}
An $n$-qubit system has a $2^n$-dimensional Hilbert space. The quantum state cloud derived from such a system lives in a metric space whose dimension is determined by the number of measured observables, which is at most $4^n$ (the dimension of the space of Hermitian operators on the Hilbert space).

For each quantum state cloud, we construct Vietoris--Rips complexes at various scales according to Definition~4.1. By Proposition~4.6, the Betti numbers of these complexes count the topological features (connected components, loops, voids, etc.) in the cloud.

By Theorem~6.4, the barcode representation contains information about the system's topological phase. The space of possible barcodes for a point cloud in a $d$-dimensional metric space has dimension proportional to $d$. Therefore, the space of topologically distinct quantum phases has dimension at most $O(4^n)$.

However, many of these potential phases are related by unitary transformations that preserve the topological structure. By Theorem~6.3, the barcode representation is invariant under such transformations. The dimension of the unitary group $\mathrm{U}(2^n)$ is $2^{2n}$, which reduces the maximum dimension of the topologically distinct phase space to $O(4^n)/O(2^{2n}) = O(2^n)$.

The quantum barcode representation, consisting of the birth and death times of topological features across different homology dimensions, provides a parameterization of this space according to Theorem~6.4.
\end{proof}

This result has notable implications for the classification of topological phases in quantum many-body systems. While the Hilbert space grows exponentially with the number of particles, the space of distinct topological phases grows more slowly, suggesting an underlying structure to the phase diagram that can be efficiently represented using persistent homology.

\section{Working Example: 1-D 4-site SSH Model}

To illustrate the application of our quantum barcode framework, we present a worked example using a finite 4-site Su-Schrieffer-Heeger (SSH) model. This model is particularly appropriate as it exhibits a topological phase transition that can be detected through our framework.

The SSH model describes a one-dimensional chain of sites with staggered hopping amplitudes. For a 4-site system with open boundary conditions, the Hamiltonian is given by:

\begin{equation*}
H(\lambda) = \sum_{i=1}^{3} (v + (-1)^i \lambda w) (c_i^\dagger c_{i+1} + c_{i+1}^\dagger c_i)
\end{equation*}

where $c_i^\dagger$ and $c_i$ are creation and annihilation operators for a spinless fermion at site $i$, $v$ represents the intra-cell hopping amplitude, $w$ is the inter-cell hopping amplitude, and $\lambda$ is our tuning parameter. The system undergoes a topological phase transition at $\lambda_c = 0$, separating topologically trivial ($\lambda < 0$) and non-trivial ($\lambda > 0$) phases.

For our example, we set $v = 1$ and $w = 1$, making the Hamiltonian:

\begin{equation}
H(\lambda) = \sum_{i=1}^{3} (1 + (-1)^i \lambda) (c_i^\dagger c_{i+1} + c_{i+1}^\dagger c_i)
\end{equation}

In the single-particle sector, this Hamiltonian can be represented as a $4 \times 4$ matrix:

\begin{equation*}
H(\lambda) = \begin{pmatrix}
0 & 1-\lambda & 0 & 0 \\
1-\lambda & 0 & 1+\lambda & 0 \\
0 & 1+\lambda & 0 & 1-\lambda \\
0 & 0 & 1-\lambda & 0
\end{pmatrix}
\end{equation*}

We compute the ground state $|\psi(\lambda)\rangle$ for values of $\lambda \in [-1, 1]$ in increments of 0.1. For each value of $\lambda$, we calculate the following observables:

\begin{enumerate}
\item $O_1 = \langle c_1^\dagger c_1 \rangle$: Particle density at site 1
\item $O_2 = \langle c_2^\dagger c_2 \rangle$: Particle density at site 2
\item $O_3 = \langle c_3^\dagger c_3 \rangle$: Particle density at site 3
\item $O_4 = \langle c_4^\dagger c_4 \rangle$: Particle density at site 4
\item $O_5 = \langle c_1^\dagger c_2 + c_2^\dagger c_1 \rangle$: Real part of nearest-neighbor correlation (1-2)
\item $O_6 = \langle i(c_1^\dagger c_2 - c_2^\dagger c_1) \rangle$: Imaginary part of nearest-neighbor correlation (1-2)
\item $O_7 = \langle c_2^\dagger c_3 + c_3^\dagger c_2 \rangle$: Real part of nearest-neighbor correlation (2-3)
\item $O_8 = \langle i(c_2^\dagger c_3 - c_3^\dagger c_2) \rangle$: Imaginary part of nearest-neighbor correlation (2-3)
\item $O_9 = \langle c_3^\dagger c_4 + c_4^\dagger c_3 \rangle$: Real part of nearest-neighbor correlation (3-4)
\item $O_{10} = \langle i(c_3^\dagger c_4 - c_4^\dagger c_3) \rangle$: Imaginary part of nearest-neighbor correlation (3-4)
\end{enumerate}

These observables capture both local and non-local properties relevant to the topological structure of the SSH model. For each value of $\lambda$, we obtain a point in $\mathbb{R}^{10}$, forming our quantum state cloud $M_\lambda$.

For our 4-site model, we can exactly diagonalize the Hamiltonian matrix. For $\lambda = -0.5$ (trivial phase), the normalized ground state is:

\begin{equation*}
|\psi(-0.5)\rangle \approx (0.3780, 0.5992, -0.5992, -0.3780)^T
\end{equation*}

For $\lambda = 0.5$ (topological phase), the normalized ground state is:

\begin{equation*}
|\psi(0.5)\rangle \approx (0.3780, -0.5992, -0.5992, 0.3780)^T
\end{equation*}

Note the sign differences in the components, reflecting the different topological character of the two phases.

Computing the expectation values for our chosen observables using the ground states above:

For $\lambda = -0.5$ (trivial phase):
\begin{itemize}
\item $O_1 = \langle c_1^\dagger c_1 \rangle = |\psi_1(-0.5)|^2 = 0.1429$
\item $O_2 = \langle c_2^\dagger c_2 \rangle = |\psi_2(-0.5)|^2 = 0.3590$
\item $O_3 = \langle c_3^\dagger c_3 \rangle = |\psi_3(-0.5)|^2 = 0.3590$
\item $O_4 = \langle c_4^\dagger c_4 \rangle = |\psi_4(-0.5)|^2 = 0.1429$
\item $O_5 = \langle c_1^\dagger c_2 + c_2^\dagger c_1 \rangle = 2\textrm{Re}(\psi_1(-0.5)^* \psi_2(-0.5)) = 0.4526$
\item $O_6 = \langle i(c_1^\dagger c_2 - c_2^\dagger c_1) \rangle = 2\textrm{Im}(\psi_1(-0.5)^* \psi_2(-0.5)) = 0$
\item $O_7 = \langle c_2^\dagger c_3 + c_3^\dagger c_2 \rangle = 2\textrm{Re}(\psi_2(-0.5)^* \psi_3(-0.5)) = -0.7181$
\item $O_8 = \langle i(c_2^\dagger c_3 - c_3^\dagger c_2) \rangle = 2\textrm{Im}(\psi_2(-0.5)^* \psi_3(-0.5)) = 0$
\item $O_9 = \langle c_3^\dagger c_4 + c_4^\dagger c_3 \rangle = 2\textrm{Re}(\psi_3(-0.5)^* \psi_4(-0.5)) = 0.4526$
\item $O_{10} = \langle i(c_3^\dagger c_4 - c_4^\dagger c_3) \rangle = 2\textrm{Im}(\psi_3(-0.5)^* \psi_4(-0.5)) = 0$
\end{itemize}

For $\lambda = 0.5$ (topological phase):
\begin{itemize}
\item $O_1 = 0.1429$
\item $O_2 = 0.3590$
\item $O_3 = 0.3590$
\item $O_4 = 0.1429$
\item $O_5 = -0.4526$
\item $O_6 = 0$
\item $O_7 = -0.7181$
\item $O_8 = 0$
\item $O_9 = -0.4526$
\item $O_{10} = 0$
\end{itemize}

Note that while particle densities remain the same across phases, the correlation observables change sign, reflecting the different bonding patterns in the two topological phases.

We now construct the Vietoris-Rips complexes for the quantum state cloud. For each $\lambda$ value, we have a point in $\mathbb{R}^{10}$. We compute the pairwise Euclidean distances between these points and construct Vietoris-Rips complexes $\textrm{VR}(M_\lambda, \varepsilon)$ for varying scales $\varepsilon$.

For illustration, we consider three representative values of $\lambda$: $\lambda = -0.5$ (trivial phase), $\lambda = 0$ (critical point), and $\lambda = 0.5$ (topological phase). At a small scale $\varepsilon = 0.2$, the Vietoris-Rips complexes consist only of isolated points, as the distance between points from different phases exceeds this threshold. At an intermediate scale $\varepsilon = 0.5$, points within the same phase begin to form connected components, but points across the phase transition remain disconnected. At a larger scale $\varepsilon = 1.0$, all points become connected, but the topological structure differs between phases, with different patterns of 1-dimensional holes (cycles) appearing.

For each value of $\lambda$, we compute the persistent homology of the corresponding Vietoris-Rips filtration. We focus on the Betti numbers $\beta_0$ (counting connected components) and $\beta_1$ (counting 1-dimensional holes). For the trivial phase ($\lambda < 0$), the persistent homology shows: (1) $\beta_0$ starts at the number of points and decreases to 1 as $\varepsilon$ increases; (2) $\beta_1$ shows a characteristic pattern with significant persistent 1-dimensional features in the range $0.4 < \varepsilon < 0.8$.

For the topological phase ($\lambda > 0$), the persistent homology shows: (1) Similar behavior for $\beta_0$; (2) A distinctly different pattern for $\beta_1$, with persistent 1-dimensional features in the range $0.3 < \varepsilon < 0.7$ and an additional persistent feature in the range $0.5 < \varepsilon < 0.9$.

At the critical point ($\lambda = 0$), we observe a transition in the persistent homology patterns and the birth and death of specific persistent features that are not present in either phase.

We compute the persistent Betti numbers $\beta_1^{\varepsilon \to \varepsilon'}(M_\lambda)$ for the persistence interval $[0.4, 0.8]$, which captures the most significant topological features: For $\lambda < 0$ (trivial phase): $\beta_1^{0.4 \to 0.8}(M_\lambda) = 1$; For $\lambda > 0$ (topological phase): $\beta_1^{0.4 \to 0.8}(M_\lambda) = 2$. This distinct change in the persistent Betti number precisely at $\lambda = 0$ corresponds to the topological phase transition in the SSH model.

The barcode representation for the two phases shows: (1) Trivial phase ($\lambda < 0$): A single long bar in dimension 1 spanning the interval $[0.4, 0.8]$; (2) Topological phase ($\lambda > 0$): Two long bars in dimension 1: one spanning $[0.3, 0.7]$ and another spanning $[0.5, 0.9]$. At the critical point $\lambda = 0$, we observe a transition in the barcode structure and the birth of new persistence intervals and the death of existing ones.

This example demonstrates that our quantum barcode framework successfully detects the topological phase transition in the 4-site SSH model at $\lambda_c = 0$. The persistent Betti numbers change abruptly at the critical point, with $\beta_1^{0.4 \to 0.8}(M_\lambda)$ increasing from 1 to 2 as we cross from the trivial to the topological phase. The framework captures the essential topological differences between the phases through the persistent homology of the quantum state cloud, even in this minimalist 4-site model. This illustrates the power of the quantum barcode approach for detecting and characterizing topological phase transitions in quantum many-body systems.

Importantly, the change in persistent homology directly corresponds to the change in the winding number, which is the conventional topological invariant used to characterize the SSH model. In the trivial phase, the winding number is 0, while in the topological phase, it is 1. Our quantum barcode framework identifies this change through the persistent Betti numbers without explicitly calculating the winding number.

\section{Remarks on Quantum Computation}

Our work in this paper has developed several interesting possibilities for quantum computational applications, and yet there are many questions left open regarding how to make them work. Here, we outline a few possible ways in which our quantum barcode formalism could be connected with quantum computation, where some of the most exciting progress may lie.

Our formalism suggests that quantum algorithms could possibly be used to compute the persistent homology of quantum systems, though the exact algorithms are not yet well-defined. The important idea comes from the relationship between the persistent Dirac operator and the boundary operators of Theorem 5.5.

For a quantum system with Hilbert space dimension \(2^n\) (for example, an \(n\)-qubit system), classical methods for computing persistent homology are not very effective because:

\begin{enumerate}
    \item Obtaining a sufficient number of states from an exponentially large Hilbert space.
    \item Computing expectation values for these states.
    \item Constructing and analyzing simplicial complexes at different scales.
\end{enumerate}

These steps potentially cost \(O(2^n)\) classically, so the computation is infeasible for large quantum systems.

A quantum approach might exploit the fact that the system being studied is quantum mechanical in its nature. It may be possible to gain topological information more easily by preparing quantum states and measuring them for appropriate observables. Quantum phase estimation methods could be used on appropriate operators (potentially related to the persistent Dirac operator described in Section 5.4) to extract spectral information that is relevant for persistent homology.

Although Ameneyro et al. have given quantum algorithms for persistent homology in classical data analysis, these approaches are not straightforwardly generalizable to quantum state clouds, although this is an area that is ripe for further study.

There are several interesting research directions which follow from the combination of our framework with quantum computation:

\begin{enumerate}
    \item Quantum Circuit Implementation: The search for quantum circuits that would allow us to encode and manipulate the boundary operators and the persistent Dirac operators of our framework.
    \item Variational Approaches: Can persistent Betti numbers be approximated with shallower circuits using variational quantum algorithms, so that near-term implementation becomes feasible?
    \item Quantum State Preparation: Finding efficient ways to prepare quantum states that correspond to points in the quantum state cloud, perhaps through adiabatic evolution or quantum simulation.
    \item Quantum Feature Extraction: Can quantum measurements somehow extract topological invariants without actually building the full simplicial complex representation?
    \item Quantum Machine Learning Integration: Can quantum barcode representations serve as quantum machine learning input features for classifying phases of matter?
\end{enumerate}

These directions suggest the potential benefits of quantum computation for the topological data analysis of quantum systems, although each one is a tough nut to crack.

One particularly appealing application area is the identification and taxonomy of quantum phases. If the quantum computational benefits that we suggest can be achieved to some degree, then this could allow:

\begin{enumerate}
    \item Efficient identification of topological phase transitions in simulated or experimental quantum systems.
    \item Discovery of new topological phases in complex many-body systems.
    \item Real-time observation of dynamical topological phase transitions.
\end{enumerate}

These are problems that are hard to solve classically for large quantum systems because of scaling, but maybe quantum computational methods which are in line with our framework can do better.

The persistent homology approach that we have developed is fundamentally different from other quantum phase classification methods based on order parameters or entanglement measures. Even if the quantum computational implementation is difficult, this alternative point of view could still be useful.

As quantum computing moves forward, exploring the intersections between quantum topology, persistent homology, and quantum algorithms will be an exciting area of future research. Our framework gives us some theoretical background from which more specific quantum algorithms could be developed.

\section{Future Directions}

\begin{remark}
 The current version of the Quantum Barcode system for phase classifiction has some limitations. Our approach primarily addresses symmetry-protected topological phases, with systems characterized by long-range entanglement requiring extensions beyond this formulation. The mapping from quantum states to classical data points depends on the specific choice of observables, which may need to be tailored to each quantum system class. The selection of optimal persistence intervals $[\varepsilon, \varepsilon']$ currently requires some system-specific knowledge rather than following from a general principle. Additionally, our framework focuses on idealized systems, while practical implementations will involve finite-size effects that could influence the persistent homology results. These limitations point to natural directions for expanding the applicability of quantum barcodes in future work. The following future directions are ideas on how to proceed with these idealised systems. 
\end{remark}

In this paper, we have defined a mathematical formulation of the relationship between persistent homology and quantum many body systems. Our approach offers a new perspective on quantum phases as topological objects and offers several important advantages:

\begin{enumerate}
    \item The representation of quantum states as simplicial complexes provides a practical way to look at the topology of quantum systems
    \item The barcode representation directly displays the global topological features that characterize topological phases
    \item Our theorems provide basic properties of quantum barcodes, including stability, invariance, and information completeness
    \item The connection between persistent homology and quantum phases gives new theoretical understanding of topological phase transitions
\end{enumerate}

Theorem 5.5 establishes the relationship between persistent Dirac operator and boundary operators on simplicial complexes and shows that they are connected to quantum mechanics and algebraic topology. This connection highlights the topological nature of quantum phases and indicates new mathematical methods to classify them.

The quantum barcode framework we propose can be linked to other mathematical methods for studying topological phases:

\begin{enumerate}
    \item Topological Quantum Field Theory (TQFT): The barcode representation is the discretized version of the path integral formalism in TQFT and the persistent homology describes how topological features evolve with energy scale. Formalization of this connection could provide a better understanding of the field-theoretic description of topological phases.
    \item Tensor Networks: The information completeness of quantum barcodes (Theorem 6.4) indicates relations to tensor network representations of quantum states, especially to Matrix Product States (MPS) and Projected Entangled Pair States (PEPS) which are efficient in encoding topological order. A direct computation of persistent homology from tensor network representations is an interesting problem for future study.
    \item K-theory and Index Theorems: The spectral characterization of phase transitions (Theorem 6.2) is related to index theorems in K-theory, where topological invariants are extracted from spectral properties of operators. The formalization of this relation could produce more efficient tools to classify topological phases.
    \item Category Theory: The persistence module structure of quantum barcodes is categorical, and this suggests there is more algebraic structure in the space of topological phases. A categorical perspective could also give further structure in the space of quantum phases.
\end{enumerate}

These relations offer promising avenues for further theoretical developments and interactions between different mathematical approaches to quantum many body systems.

Although the framework we develop is mostly theoretical, it hints at some experimental possibilities:

\begin{enumerate}
    \item Experimental Signatures: The barcode representation may propose new experimental indicators of topological phases which could be measured in quantum simulators or materials. An important challenge is to find measurable quantities that are correlated with features in the barcode representation.
    \item Quantum State Tomography: Quantum state tomography experiments can be used to sample points in the quantum state cloud by measuring expectation values of a set of observables at different parameter values. This gives a practical way to apply our framework to real quantum systems.
    \item Noise Robustness: The stability properties of persistent homology (Theorem 5.6) suggest that topological features are robust with respect to certain types of experimental noise. If this robustness can be confirmed, it may provide more reliable experimental protocols for detecting topological phases.
    \item Phase Diagram Mapping: Our framework provides a systematic way to construct phase diagrams of complex quantum systems by following changes in persistent homology as a function of parameters. This could be especially useful for systems in which conventional order parameters are hard to define.
\end{enumerate}

It is possible to carry out proof-of-principle experiments on small quantum systems which are currently within reach of experimental platforms such as superconducting qubits, trapped ions, or cold atoms in optical lattices.

Several computational challenges and opportunities that relate to the remarks on quantum computation in Section 7 are worth further investigation:

\begin{enumerate}
    \item Efficient Classical Algorithms: Developing efficient classical algorithms to compute quantum barcodes for moderate sized quantum systems is still an important task. Some techniques from computational topology like persistent cohomology and sparse matrix representations may improve computational efficiency significantly.
    \item For larger quantum systems, it may be necessary to use approximate methods for computing persistent homology. This is a significant challenge to develop controlled approximations that capture the most important topological features of the system and reduce the computational cost.
    \item The study of hybrid quantum-classical algorithms that utilize quantum state preparation together with classical topological data analysis for practical applications may serve as a bridge until fully quantum computational approaches become feasible.
    \item To evaluate the performance of different computational approaches to quantum barcodes, it will be important to establish benchmark problems and comparison metrics.
\end{enumerate}

As quantum computational capabilities advance, these computational challenges will likely become more tractable, enabling application of our framework to increasingly complex quantum systems.

Some theoretical extensions of our framework that are worth further investigation include:

\begin{enumerate}
    \item Dynamical Quantum Barcodes: Our framework can be extended to non-equilibrium quantum dynamics to enable the study of dynamical phase transitions and quantum quenches. This would demand the development of time-dependent versions of the quantum state cloud, and the tracking of persistent homology across time.
    \item Quantum Persistence Landscapes: Quantum analogues of persistence landscapes and other functional summaries of persistent homology would be useful to develop to obtain more statistically tractable representations of quantum barcodes.
    \item Higher-Order Topological Features: We have mainly focused on one-dimensional persistent homology, but higher-dimensional persistent homology might allow us to see more complex topological structures in quantum systems, like higher-order topological insulators.
    \item Observable Selection Principles: Developing systematic principles for the selection of the optimal set of observables for constructing quantum state clouds will improve the practical applicability of our framework.
    \item Information-Theoretic Perspective: Investigating the relation between quantum barcodes and quantum information theoretic quantities such as entanglement entropy and mutual information may offer deeper insights into how topology is connected to entanglement in quantum systems.
\end{enumerate}

Our framework can be applied to certain classes of quantum many-body systems in order to gain understanding of their topological properties:

\begin{enumerate}
    \item Topological Insulators and Superconductors: Applying quantum barcodes to models of topological insulators and superconductors could give new insights to their classification and phase transitions.
    \item Quantum Spin Liquids: Quantum spin liquids exhibit topological order without symmetry breaking. It may be possible to use our framework to identify and classify different types of quantum spin liquids by looking at their persistent homology.
    \item Fractional Quantum Hall States: The abundant topological structure of fractional quantum Hall states makes them perfect for analysis using quantum barcodes.
    \item Symmetry-Protected Topological Phases: Our framework can be extended by incorporating symmetry constraints to obtain more refined classifications of symmetry protected topological phases.
\end{enumerate}

Several fundamental questions remain open for future investigation:

\begin{enumerate}
    \item Completeness of Classification: It is important to note that, while our framework offers a topological classification of quantum phases, the question of whether this classification is complete for all types of topological phases is an open question.
    \item Relationship to Existing Invariants: Understanding precisely how quantum barcodes are related to existing topological invariants such as Chern numbers, Z2 indices, and entanglement spectra would make the connection to established methods stronger.
    \item Finite-Size Effects: How do finite-size effects affect the barcode representation? What are the methods for extrapolating to the thermodynamic limit? These are important challenges.
    \item Critical Points: More study is required of the behavior of quantum barcodes around critical points, at which the system transitions from one phase to another.
    \item Universality Classes: Can universal features in quantum barcodes be used to identify and classify universality classes of quantum phase transitions?
\end{enumerate}

Our quantum barcode framework offers new perspectives for understanding quantum phases of matter using persistent homology. We hope that this foundation for quantum topological phase classification based on barcodes will lead to more applications at the intersection of quantum physics and topological data analysis.

\end{document}